\documentclass[letterpaper, 10 pt, conference]{ieeeconf}

\IEEEoverridecommandlockouts
\usepackage{url} 
\usepackage{algorithm}
\usepackage{algorithmic}
\usepackage{amsmath,amssymb,amsfonts}
\usepackage[font=small,labelfont=bf]{caption}
\usepackage{epsfig,graphics}
\usepackage{xcolor}
\usepackage{cite}

\usepackage{pgf}
\usepackage{tikz}
\usetikzlibrary{arrows,automata}

\graphicspath{{./fig/}}
\usepackage{hhline}

\usepackage{theorem}
\newtheorem{lemma}{Lemma}
\newtheorem{remark}{Remark}

\newtheorem{assumption}{Assumption}
\newtheorem{corollary}{Corollary}
\newtheorem{definition}{Definition}

\newtheorem{problem}{Problem}

\newcommand{\B}{\mathcal{B}} 

\newcommand{\N}{\mathcal{N}} 

\newcommand{\TimeA}{\mathbb{T}^A_i}
\newcommand{\TimeS}{\mathbb{T}^S_i}
\newcommand{\Ser}{\mathbb{S}}
\newcommand{\timeS}{\tau}
\newcommand{\timeA}{\mathit{t}}
\newcommand{\Act}{\mathbb{A}}
\newcommand{\traj}{\mathbf{x}}
\newcommand{\trace}{\mathit{trace}}
\newcommand{\leader}{\ell}
\newcommand{\plan}{\Omega}
\newcommand{\urge}{\Upsilon}
\newcommand{\last}{\lambda}
\renewcommand{\next}{\nu}
\newcommand{\curr}{{t}_\mathit{curr}}

\newcommand{\AP}{\Pi} 
\newcommand{\Reg}{\mathcal R} 
\newcommand{\reg}{R} 
\renewcommand{\min}{\mathit{min}}
\newcommand{\com}{{}}

\newcommand{\Lang}{\mathcal{L}} 
\newcommand{\Set}{\mathsf{S}} 

\newcommand{\C}{\mathcal{C}}

\newcommand{\Next}{\mathsf{X}}
\newcommand{\Until}{\mathsf{U}}
\newcommand{\Always}{\mathsf{G}}
\newcommand{\Event}{\mathsf{F}}

\renewcommand{\epsilon}{\varepsilon}

\newcommand{\eg}{{e.g., }}

\renewcommand{\alph}{\Sigma}

\begin{document}
\title{\textbf {Cooperative Decentralized Multi-agent Control under Local LTL Tasks and Connectivity Constraints}}
\author{Meng Guo, Jana T\r{u}mov\'a and Dimos V. Dimarogonas \thanks{The authors are with the ACCESS Linnaeus Center, School of Electrical
Engineering, KTH Royal Institute of Technology, SE-100 44, Stockholm,
Sweden and with the KTH Centre for
Autonomous Systems. \texttt{mengg, tumova, dimos@kth.se}. This work was supported by the EU STREP RECONFIG: FP7-ICT-2011-9-600825.}}
\maketitle

\begin{abstract}We propose a framework for the decentralized control of a team of agents that are assigned local tasks expressed as Linear Temporal Logic (LTL) formulas. Each local LTL task specification captures both the requirements on the respective agent's behavior and the requests for the other agents' collaborations needed to accomplish the task. Furthermore, the agents are subject to {communication} constraints. The presented solution follows the automata-theoretic approach to LTL model checking, however, it avoids the computationally demanding construction of synchronized product system between the agents. We suggest a decentralized coordination among the agents through a dynamic leader-follower scheme, to guarantee the low-level connectivity maintenance at all times and a progress towards the satisfaction of the leader's task. By a systematic leader switching, we ensure that each agent's task will be accomplished. 
\end{abstract}

\section{Introduction}

Cooperative control for multi-agent systems have been extensively studied for various purposes like consensus~\cite{1470239}, formation~\cite{1641830},~\cite{egerstedt2001formation}, and reference-tracking~\cite{hong2006tracking}, where each agent either serves to accomplish a global objective or fulfil simple local goals such as reachability. In contrast, we focus on planning under complex tasks assigned to the agents, such as periodic surveillance (repeatedly perform $A$), sequencing (perform $A$, then $B$, then $C$), or request-response (whenever $A$ occurs, perform $B$). Particularly, we follow the idea of correct-by-design control from temporal logic specifications that has been recently largely investigated both in single-agent~and multi-agent settings.
In particular, we consider a team of agents modeled as a dynamical system that are assigned a local task specification as Linear Temporal Logic (LTL) formulas. The agents might not be able to accomplish the tasks by themselves and hence requirements on the other agents' behaviors are also part of the LTL formulas. Consider for instance a team of robot operating in a warehouse that are required to move goods between certain warehouse locations. While light goods can be carried by a single robot, help from another robot is needed to move heavy goods, i.e. the requirement on another agents' behavior is a part of its LTL task specification.

{The goal of this work is to find motion controllers and action plans for the agents that guarantee the satisfaction of all individual LTL tasks. We aim for a decentralized solution while taking into account the constraints that the agents can exchange messages only if they are close enough. Following the hierarchical approach to LTL planning, we first generate for each agent a sequence of actions as a high-level plan that, if followed, guarantees the accomplishment of the respective agent's LTL task. Second, we merge and implement the syntesized plans in real-time, upon the run of the system. Namely, we introduce a distributed continuous controller for the leader-follower scheme, where the current leader guides itself and the followers towards the satisfaction of the leader's task. At the same time, the connectivity of the multi-agent system is maintained.  By a systematic leader re-election, we ensure that each agent's task will be met in long term.}

Multi-agent planning under temporal logic tasks has been studied in several recent papers~\cite{quo-icra2004,loizou-cdc2005,marius-cdc2011,sertac-ijnc2010,yushan-tr2012,alphan-ijrr2013,meng-cdc2013,lygeros-ecc2013, jana-acc2014}. Many of them build on top-down approach to planning, when a single LTL task is given to the whole team. For instance, in~\cite{yushan-tr2012,alphan-ijrr2013}, the authors propose decomposition of the specification into a conjunction of independent local LTL formulas. On the other hand, we focus on bottom-up planning from individual specification. Related work includes a decentralized control of a robotic team from local LTL specification with communication constraints proposed in~\cite{dimos-cdc12}. However, the specifications there are truly local and the agents do not impose any requirements on the other agents' behavior. 
In~\cite{meng-cdc2013}, the same bottom-up planning problem from LTL specifications is considered and a partially decentralized solution is designed that takes into account only clusters of dependent agents instead of the whole group. This approach is later extended in~\cite{jana-acc2014}, where a receding horizon approach to the problem is suggested. Both mentioned studies however assume that the agents are fully synchronized in their discrete abstractions and the proposed solutions rely on construction of the synchronized product system between the agents, or at least of its part. In contrast, in this work, we avoid the product  construction completely.

The contribution of the paper can be summarized as the proposal of a decentralized motion and action control scheme for multi-agent systems with complex local tasks which handles both connectivity constraints and collaborative tasks. The features of the suggested solution are as follows: (1) the continuous controller is distributed and integrated with the leader election scheme; (2) the distributed leader election algorithm only requires local communications and guarantees sequential progresses towards individual desired tasks; and (3) the proposed coordination scheme operates in real-time, upon the run of the system as opposed to offline solutions that require fully synchronized motions of all agents. 

The rest of the paper is organized as follows. In Section~\ref{sec:prelims} we state the necessary preliminaries. Section~\ref{sec:pf} formally introduces the considered problem. In Section~\ref{sec:ps} we describe the proposed solution in details. Section~\ref{sec:example} demonstrates the results in a simulated case study. Finally, we conclude in Section~\ref{sec:conc}.

\section{Preliminaries}
\label{sec:prelims}

Given a set $\Set$, let 
$2^\Set$,
and $\Set^\omega$
denote 
the set of all subsets of $\Set$, and the set of all infinite sequences of elements of $\Set$, respectively.
An infinite sequence of elements of $\Set$ is called an infinite word over $\Set$, respectively. 

{\begin{definition}
An LTL formula $\phi$ over the set of services $\Sigma$ is defined
  inductively as follows:
   \begin{enumerate}
  \setlength{\itemsep}{1pt}
  \setlength{\parskip}{0pt}
  \setlength{\parsep}{0pt}
  \item every service $\sigma \in \Sigma$ is a formula, and
  \item if $\phi_1$ and $\phi_2$ are formulas, then $\phi_1 \lor
    \phi_2$, $\lnot \phi_1$, $\Next\, \phi_1$, $\phi_1\,\Until\,\phi_2$, $\Event \, \phi_1$, and $\Always \, \phi_1$
    are each formulas,
  \end{enumerate}
 where $\neg$ (negation) and $\vee$
  (disjunction) are standard Boolean connectives, and $\Next$ (next), $\Until$ (until), $\Event$ (eventually), and  $\Always$ (always) are temporal operators.
  \end{definition}}

The semantics of LTL is defined over infinite words over~$2^\Sigma$. Intuitively, $\sigma$ is satisfied on a word $w = w(1)w(2)\ldots$ if it holds at its first position $w(1)$, i.e. if $\sigma \in w(1)$. Formula $\Next \, \phi$ holds true if $\phi$ is satisfied on the word suffix that begins in the next position $w(2)$, whereas $\phi_1 \, \Until\, \phi_2$ states that $\phi_1$ has to be true until $\phi_2$ becomes true. Finally, $\Event \, \phi$ and $\Always \, \phi$ are true if $\phi$ holds on $w$ eventually, and always, respectively. For the formal definition of the LTL semantics see, e.g.~\cite{principles}.

The set of all words that are accepted by an LTL formula $\phi$ is denoted by $\Lang(\phi)$. 

\begin{definition}[B\"uchi Automaton]
A B\"uchi automaton over alphabet $2^\Sigma$ is a tuple $\B =  (Q,q_{init},2^\Sigma,\delta,F)$, where
\begin{itemize}
\item 
$Q$ is a finite set of states; 
\item 
$q_{init}\in Q$ is the initial state; 
\item 
$2^\Sigma$ is an input alphabet; 
\item 
$\delta \subseteq Q \times \Sigma \times Q$ is a non-deterministic transition relation; 
\item 
$F$ is the acceptance condition.
\end{itemize}
\end{definition}

The semantics of B\"uchi automata are defined over infinite input words over $2^\Sigma$.
A \emph{run} of the B\"uchi automaton $\B$ \emph{over} an input word $w=w(1)w(2)\ldots$  is a sequence
$\rho=q_1q_2\ldots$, such that $q_1  = q_{init}$, and
$(q_i,w(i),q_{i+1}) \in \delta$, for all $i\geq 1$.
A run $\rho=q_1q_2\ldots$ is \emph{accepting} if it intersects $F$ infinitely many times. A word $w$ is \emph{accepted} by $\B$ if there exists an accepting run over $w$. The \emph{language} of all words accepted by $\B$ is denoted by $\Lang(\B)$.
Any LTL formula $\phi$ over $\Pi$ can be algorithmically translated into a B\"uchi automaton $\B$, such that $\Lang(\B) = \Lang(\phi)$~\cite{principles} and many software tools for the translation exist, e.g., ~\cite{ltl2ba}.

Given an LTL formula $\varphi$ over $\Sigma$, a word that satisfies $\varphi$ can be generated as follows. First, the LTL formula is translated into a corresponding B\"uchi automaton. Second, the B\"uchi automaton is viewed as a graph $G = (V, E)$, where $V = Q$, and $E$ is given by the transition relation $\delta$ in the expected way: $(q,q') \in E \iff \exists \Set \subseteq \Sigma$, such that $(q,\Set,q') \in \delta$. By finding a finite path (prefix) followed by a cycle (suffix) containing an accepting state, we find a word that is accepted by $\B$,  which is a word that satisfies $\varphi$ in a prefix-suffix form $\Set_1\ldots \Set_p(\Set_{p+1}\ldots \Set_s)^\omega$. Details can be found e.g., in~\cite{principles}.

In this particular work, we are interested only in subsets of $2^\Sigma$ that are singletons. Thus, with a slight abuse of notation, we interpret LTL over words over $\Sigma$, i.e. over sequences of services instead of sequences of subsets of services.

\section{Problem Formulation}\label{formulation}
\label{sec:pf}
\subsection{Agent Dynamics and Network Structure}\label{system}
Let us consider a team of $N$ agents, modeled by the single-integrator dynamics:
\begin{equation}\label{dynamics}
\dot{x}_i(t) = u_i(t), \qquad i\in \N=\{1,\ldots, N\},
\end{equation}
where $x_i(t), \, u_i(t) \in \mathbb{R}^2$ are the state and control inputs of agent~$i$ at time {$t> 0$, $x_i(0)$ is the given initial state, and $\traj_i(t)$ is the \emph{trajectory} of agent $i$ from $0$ to $t\geq 0$. We assume that all agents start at the same instant $t=0$.}

Suppose that each of the agents has a limited communication radius of $r_\com>0$. This means that at time $t$, agent $i$ can communicate, i.e., exchange messages directly with agent $j$ if and {only if} $\|x_i(t)-x_{j}(t)\|\leq r_{\com}$. This constraint imposes certain challenges on the distributed coordination of multi-agent systems as the inter-agent communication or information exchange depends on their relative positions. 

Agents $i$ and $j$ are \emph{connected} at time $t$ if and only if either $\|x_i(t)-x_{j}(t)\|\leq r_{\com}$, or if there exists $i'$, such that  $\|x_i(t)-x_{i'}(t)\|\leq r_{\com}$, where $i'$ and $j$ are connected. Hence, two connected agents can communicate indirectly. We assume that initially, all agents are connected.
{The particular message passing protocol is beyond the scope of this paper. For simplicity, we assume that message delivery is reliable, meaning that a message sent by agent $i$ will be received by all connected agents $j$. }

\subsection{Task Specifications}
	Each agent $i \in \N$ is assigned a set of $M_i$ \emph{services} $\Sigma_i = \{\sigma_{ih}, h \in \{1,\ldots,M_i\}\}$ that it is responsible for, and a set of $K_i$
regions, where subsets of these services can be \emph{provided}, denoted by $\Reg_i=\{\reg_{ig},\, g \in \{1,\ldots,K_i\}\}$. 
For simplicity of presentation, $\reg_{ig}$ is determined by a circular area: 
\begin{equation}\label{regions}
\reg_{ig}=\{y\in \mathbb{R}^2|\|y-c_{ig}\|\leq r_{ig}\}
\end{equation}
where $c_{ig}\in \mathbb{R}^2$ and $r_{ig}$ are the center and radius of the region, respectively, such that $r_{ig} \geq r_\min>0$, for a fixed minimal radius $r_\min$. {Furthermore, each region in $\Reg_i$ is reachable for each agent.}
Labeling function $L_i: \Reg_i \to 2^{\Sigma_i}$ assigns to each region $R_{ig}$ the set of services $L_i(R_{ig}) \subseteq \Sigma_i$ that can be provided in there.

Some of the services in $\Sigma_i$ can be provided solely by the agent $i$, while others require cooperation with some other agents. Formally, agent $i$ is associated with a set of \emph{actions} $\Pi_i$ that it is capable of \emph{executing}. The actions are of two types:
\begin{itemize}
\item action $\pi_{ih}$ of \emph{providing} the service $\sigma_{ih} \in \Sigma_i$; and
\item action $\varpi_{ii'h'}$ of \emph{cooperating} with the agent $i'$ in providing its service $\sigma_{i'h'} \in \Sigma_{i'}$.
\end{itemize}

A service $\sigma_{ih}$ then takes the following form:
\begin{equation}
\sigma_{ih} = \pi_{ih} \wedge \bigwedge_{i' \in \C_{ih}} \varpi_{i'ih},
\end{equation}
for the set of cooperating agents $\C_{ih}$, where $\emptyset \subseteq \C_{ih} \subseteq \N \setminus \{i\}$. Informally, a service $\sigma_i$ is provided if the agent's relevant service-providing action and the corresponding cooperating agents' actions are executed at the same time. Furthermore, it is required that at the moment of service providing, the agent and the cooperating agents from $\C_{ih}$ occupy the same region $R_{ig}$, where $\sigma_{ih} \in L_i(R_{ig})$.

\begin{definition}[Trace] A valid trace of agent $i$ is a tuple $\trace_i=(\traj_i(t), \TimeA, \Act_i, \TimeS, \Ser_i)$, where
\begin{itemize}
\item  $\traj_i(t)$ is a trajectory of agent $i$;
\item $\TimeA = \timeA_1,\timeA_2,\timeA_3,\ldots$ is the sequence of time instances when agent $i$ executes actions from $\AP_i$;
\item $\Act_i: \TimeA \to \AP_i$ represents the sequence of executed actions, both the service-providing and the cooperating ones; 
\item $\TimeS = \timeS_{1},\timeS_{2},\timeS_{3},\ldots$ is a sequence of time instances when services from $\Sigma_i$ are provided. Note that $\TimeS$ is a subsequence of $\TimeA$ and it is equal to the time instances when service-providing actions are executed; and

\item $\Ser_i: \TimeS \to \Sigma_i$ represents the sequence of provided services that satisfies the following property for all $l \geq 1$: There exists $g \in \{1,\ldots, K_i\}$, such that
\begin{itemize}
\item[(i)] $\traj_i(\timeS_l) \in R_{ig}$, $\Ser_i(\timeS_l) \in L_i(R_{ig})$, and $\Ser_i(\timeS_l) = \sigma_{ih} \Rightarrow \Act_i(\timeS_l) = \pi_{ih}$, and
\item[(ii)] for all $i' \in \C_{ih}$, it holds that $\traj_{i'}(\timeS_l) \in R_{ig}$ and $\Act_{i'}(\timeS_l) = \varpi_{i'ih}$.
\end{itemize}
\end{itemize}
\end{definition}

In other words, the agent $i$ can provide a service $\sigma_{ih}$ only if (i) it is present in a region $R_{ig}$, where this service can be provided, and it executes the relevant service-providing action $\pi_{ih}$ itself, and (ii) all its cooperating agents from $\C_{ih}$ are present in the same region $R_{ig}$ as agent $i$ and execute the respective cooperative actions needed.

\begin{definition}[LTL Satisfaction]
A valid trace $\trace_i=(\traj_i(t), \TimeA,\Act_i,  \TimeS= \timeS_{1},\timeS_{2},\timeS_{3},\ldots, \Ser_i: \TimeS \to \Sigma_i)$,  \emph{satisfies} an LTL formula over $\varphi_i$, denoted by $\trace_i \models \varphi_i$ if and only if $\Ser_i(\timeS_1)\Ser_i(\timeS_2)\Ser_i(\timeS_3)\ldots \models \varphi_i$.
\end{definition}

\begin{remark}
Traditionally, LTL is defined over the set of atomic propositions (APs) instead of services (see, e.g.~\cite{principles}). Usually APs represent inherent properties of system states. The labeling function $L$ then partitions APs into those that are true and false in each state. The  LTL formulas are interpreted over trajectories of systems or their discrete abstractions. 

In this work, we consider an alternative definition of LTL semantics to describe the desired tasks. Particularly, we perceive atomic propositions as offered services rather than undetachable inherent properties of the system states. For instance, given that a state is determined by the physical location of an agent, we consider atomic propositions of form ``in this location, an object can be loaded'', or ``there is a recharger in this location" rather than ``this location is dangerous''. In other words, the agent is in our case given the option to decide whether an atomic proposition $\sigma_{ih} \in L(R_{ig})$ is in state $x_i(t) \in R_{ig}$ satisfied or not. In contrast, $\sigma_i \in \Sigma_i$ is never satisfied in state $x_i(t) \in R_{ig}$, such that $\sigma_i \not \in L(R_{ig})$. The LTL specifications are thus interpreted over the sequences of provided services along the trajectories instead of the trajectories themselves.
\end{remark}

\subsection{Problem statement}
Given the above settings, we now formally state our problem:
\begin{problem}
Given a team of the {agents $\N$ subject to dynamics in Eq.~\ref{dynamics}}, synthesize for each agent $i \in \N$
\begin{itemize}
\item a control input $u_i$
\item a time sequence $\TimeA$, and 
\item an action sequence $\Act_i$,
\end{itemize}
such that the trace $\trace_i =(\traj_i(t), \TimeA, \TimeS, \Act_i, \Ser_i)$ is valid and satisfies the given local LTL task specification $\varphi_i$ over the set of services $\Sigma_i$.
\label{prob:main}
\end{problem}

\section{Problem  Solution}
\label{sec:ps}

{Our approach to the problem involves an offline and an online step. In the offline step, we synthesize a high-level plan in the form of a sequence of services for each of the agents. In the online step, we dynamically switch between the high-level plans through leader election. The whole team then follows the leader towards providing its next service.}

In this section, we provide the details of the proposed solution. Namely, we define the notion of connectivity graph for the multi-agent system as a necessary condition for the rest of the solution. Further, we focus on decentralized control of the whole team of agents towards a selected goal region $R_{\ell g}$ that is known only to a leading agent $\ell$ while maintaining their connectivity. Finally, we discuss the election of leading agents and progressive services to be provided, and goal regions to be visited that guarantee the satisfaction of all agents' tasks in long term.

\medskip

\subsection{Connectivity Graph}\label{connectivity}

Before discussing the structure of the proposed solution, let us introduce the notion of agents' connectivity graph that will allow us to handle the constraints imposed on communication between the agents.

Recall that each agent has a limited communication radius $r_\com>0$ as defined in Section~\ref{system}. Moreover, let $\varepsilon\in (0,\,r_\com)$ be a given constant. It is worth mentioning that $\varepsilon$ plays an important role for the edge definition below. In particular, it introduces a hysteresis in the definition for adding new edges to the communication graph.  
\begin{definition}\label{edge}
Let $G(t)=(\N, E(t))$ denote the undirected time-varying connectivity graph formed by the agents, where $E(t)\subseteq \N\times \N$ is the edge set  for $t \geq 0$. At time $t=0$, we set {$E(0)=\{(i,\,j)|\|x_i(0)-x_j(0)\|< r_\com \}$} 
At time $t > 0$, $(i,\, j)\in E(t)$ if and only if one of the following conditions hold:
\begin{itemize}
\item[(i)]$\|x_i(t)-x_j(t)\|\leq r_\com -\varepsilon $, or
\item[(ii)] $r_\com -\varepsilon<\|x_i(t)-x_j(t)\|\leq r_\com $ and $(i,j) \in E(t^-)$, where $t^-<t$ and $|t-t^-|\rightarrow 0$.
\end{itemize}
\end{definition}

Note that the condition (ii)  in the above definition guarantees that a new edge will only be added when the distance between two unconnected agents decreases below $r_\com -\varepsilon$. This property is crucial in proving the connectivity maintenance by Lemma~\ref{connected} and the convergence by Lemma~\ref{convergence}.

Consequently, each agent $i\in \N$ has a time-varying set of neighbouring agents, with which it can communicate directly, denoted by ${\N}_i(t)=\{i'\in \N\,|\,(i,\,i')\in E(t)\}$. 
Note that if $j$ is reachable from $i$ in $G(t)$ then agents $i$ and $j$ are connected, i.e., they can communicate directly or indirectly. From the initial connectivity requirement, we have that $G(0)$ is connected. Hence, maintaining $G(t)$ connected for all $t \geq 0$ ensures that the agents are always connected, too.

\subsection{Continuous Controller Design}

In this section, let us firstly focus on the following problem: given a leader $\leader\in \N$ at time $t$ and a goal region $R_{\ell g}\in \mathcal{R}_{\ell}$, propose a decentralized continuous controller that:  (1) guarantees that all agents $i\in\N $ reach $R_{\ell g}$ at a finite time $\overline{t}<\infty$; (2) $G(t')$ remains connected for all $t' \in [t,\,\overline{t}]$. Both objectives are critical for the leader selection scheme introduced in Section~\ref{goal-region}, which ensures sequential satisfaction of $\varphi_i$ for each $i \in \N$. 

Denote by $x_{ij}(t)=x_i(t)-x_j(t)$  the pairwise relative position between neighbouring agents, $\forall (i,\,j)\in E(t)$. Thus $\|x_{ij}(t)\|^2=\big(x_i(t)-x_j(t)\big)^T\big(x_i(t)-x_j(t)\big)$ denotes the corresponding distance.  We propose the continuous controller with the following structure:
\begin{equation}\label{control}
u_i(t)=-b_i\big(x_i-c_{ig}\big)-\sum_{j\in \N_i(t)}\nabla_{x_i}\phi\big(\|x_{ij}\|\big),
\end{equation}
where $\nabla_{x_i} \phi(\cdot)$ is the gradient of the potential function $\phi\big(\|x_{ij}\|\big)$ with respect to $x_i$, which is to be defined; $b_i\in \{0,\,1\}$ indicates if agent $i$ is the leader; $c_{ig}\in \mathbb{R}^2$ is the center of the next goal region for agent $i$; $b_i$ and $c_{ig}$ are derived from the leader selection scheme in Section~\ref{goal-region} later. 

The potential function $\phi(\|x_{ij}\|)$ is defined as follows 
\begin{equation}\label{potential1}
\phi\big(\|x_{ij}\|\big)=\frac{\| x_{ij}\|^2}{r^2-\| x_{ij}\|^2}, \qquad \|x_{ij}\|\in [0,\, r),
\end{equation} 
and has the following properties: (1) its partial derivative of $\phi(\cdot)$ over $\|x_{ij}\|$ is given by 
\begin{equation}
\begin{split}
&\frac{\partial\, \phi\big(\|x_{ij}\|\big)}{\partial\, \|x_{ij}\|}=\frac{-2r^2\, \|x_{ij}\|}{(r^2-\|x_{ij}\|^2)^2}\geq 0
\end{split}
\end{equation}
for $\|x_{ij}(t)\|\in [0,\, r)$ and the equality holds when $\|x_{ij}\|=0$; (2) $\phi\big(\|x_{ij}\|\big)\rightarrow 0$ when $\| x_{ij}\|\rightarrow 0$; (3) $\phi\big(\|x_{ij}\|\big)\rightarrow +\infty$ when $\| x_{ij}\|\in [0,\, r)$. As a result, controller~\eqref{control} becomes
\begin{equation}\label{control2}
u_i(t)=-b_i\big(x_i-c_{ig}\big)-\sum_{j\in \N_i(t)}\frac{2r^2}{(r^2-\| x_{ij}\|^2)^2}(x_i-x_j),
\end{equation}
which is fully distributed as it only depends $x_i$ and $x_j$, $\forall j\in \N_i(t)$.


\begin{lemma}\label{connected}
Assume that $G(t)$ is connected at $t=T_1$ and agent $\ell\in\N$ is the fixed leader for all $t\geq T_1$. By applying the controller in Eq.~\eqref{control2}, $G(t)$ remains connected and $E(T_1)\subseteq E(t)$ for $t\geq T_1$. 
\end{lemma}
\begin{proof}
Assume that $G(t)$ remains \emph{invariant} during $[t_1,\, t_2)\subseteq [T_1,\,\infty)$, i.e., no new edges are added to $G(t)$. Consider the following  function:
\begin{equation}\label{lyapunov}
V(t)=\frac{1}{2}\sum_{i=1}^{N}\sum_{j\in \N_i(t)}\phi(\|x_{ij}\|) + \frac{1}{2}\sum_{i=1}^{N} b_i (x_i-c_{ig})^T(x_i-c_{ig}),
\end{equation}
which is positive semi-definite. The time derivative of~\eqref{lyapunov} along system~\eqref{dynamics} is given by 
\begin{equation}
\begin{split}
\dot{V}(t)&=\sum_{i=1}^{N} \frac{\partial V}{\partial x_i}\,\dot{x}_i\\
&=\sum_{i=1,\, i\neq \ell}^{N} \bigg(\big(\sum_{j\in \N_i(t)} \nabla_{x_i}\phi(\|x_{ij}\|)\big)\,u_i\bigg)\\
&\qquad  + \bigg(\sum_{j\in \N_{\ell}(t)} \nabla_{x_{\ell}}\phi(\|x_{\ell j}\|)+ (x_\ell-c_{\ell g})\bigg)\,u_{\ell}.
\end{split}
\end{equation}
By~\eqref{control}, for follower $i\neq \ell$, the control input is given by 
$$
u_i=-\sum_{j\in \N_i(t)} \nabla_{x_i}\phi(\|x_{ij}\|)
$$
since $b_i=0$ for all followers. For the single leader $\ell$, its control input is given by 
$$
u_\ell=- (x_\ell-c_{\ell g})-\sum_{j\in \N_{\ell }(t)} \nabla_{x_{\ell}}\phi(\|x_{\ell j}\|)
$$
since $b_\ell=1$. This implies that 
\begin{equation}\label{derivative}
\begin{split}
\dot{V}(t)&=-\sum_{i=1,\, i\neq \ell}^{N} \|\sum_{j\in \N_i(t)} \nabla_{x_i}\phi(\|x_{ij}\|)\|^2\\
&\qquad  -\| (x_\ell-c_{\ell g})+\sum_{j\in \N_{\ell}(t)} \nabla_{x_{\ell}}\phi(\|x_{\ell j}\|)\,\|^2\leq 0.
\end{split}
\end{equation}
Thus $V(t)\leq V(0)<+\infty$ for $t\in [t_1,\,t_2)$. It means that during $[t_1,\, t_2)$, no existing edge can have a length close to $r$, i.e., no existing edge will be \emph{lost} by the definition of an edge. 

On the other hand, assume a \emph{new} edge $(p,\,q)$ is added to $G(t)$ at $t=t_2$, where $p,\, q\in \N$. By Definition~\ref{edge}, it holds that $\|x_{pq}(t_2)\|\leq r-\varepsilon$ and $\phi(\|x_{pq}(t_2)\|)=\frac{r-\varepsilon}{\varepsilon(2r-\varepsilon)}<+\infty$ since $0<\varepsilon<r$. Denote the set of newly-added edges at $t=t_2$ as $\widehat{E}\subset \N\times \N$. Let $V(t_2^+)$ and $V(t_2^-)$ be the value of Lyapunov function from~\eqref{lyapunov} before and after adding the set of new edges to $G(t)$ at $t=t_2$. We get
\begin{equation}
\begin{split}
V(t_2^+) &= V(t_2^-) + \sum_{(p,\,q)\in \widehat{E}}\phi(\|x_{pq}(t_2)\|)\\
&\leq  V(t_2^-) + |\widehat{E}|\, \frac{r-\varepsilon}{\varepsilon(2r-\varepsilon)}<+\infty.
\end{split}
\end{equation}
Thus $V(t)<\infty$ also holds when new edges are added. As a result, $V(t)<+\infty$ for $t\in [T_1,\, \infty)$. By Definition~\ref{edge}, one existing edge $(i,\,j)\in E(t)$ will be lost only if $x_{ij}(t)=r$. It implies that $\phi(\|x_{ij}\|)\rightarrow +\infty$, i.e., $V(t)\rightarrow +\infty$ by~\eqref{lyapunov}. By contradiction,  we can conclude that new edges might be added but no existing edges will be lost, namely $E(T_1)\subseteq E(t)$, $\forall t\geq T_1$. 

To conclude, given a connected $G(t)$ at $t=T_1$ and a fixed leader $\ell\in \N$ for $t\geq T_1$, it is guaranteed that $G(t)$ remains connected, $\forall t\geq T_1$.
\end{proof}
\begin{lemma}\label{convergence}
Given that $G(t)$ is connected at $t=T_1$ and the fixed leader $\ell \in \N$ for $t\geq T_1$, it is guaranteed that under controller in Eq.~\eqref{control2} there exist $T_1\leq \overline{t}<+\infty$
\begin{equation}
x_i(\overline{t}) \in R_{\ell g}, \qquad \forall i\in \N.
\end{equation}
\end{lemma}
\begin{proof}
First of all, it is shown in Lemma~\ref{connected} that  $G(t)$ remains connected for $t\geq T_1$ if $G(T_1)$ is connected. Moreover $E(T_1)\subseteq E(t)$, $\forall t\geq T_1$, i.e., no existing edges will be lost.

Now we show that all agents converge to the goal region of the leader in finite time. By~\eqref{derivative}, $\dot{V}(t)\leq 0$ for $t\geq T_1$ and $\dot{V}(t)=0$ when the following conditions hold: (1) for $i \neq \ell$ and $i\in \N$, it holds that 
\begin{equation}\label{condflow}
\begin{split}
&\sum_{j\in \mathcal{N}_i(t)}\frac{2r^2}{(r^2-\| x_{ij}\|^2)^2}(x_i-x_j)=0;
\end{split}
\end{equation}
(2) for the leader $\ell \in \N$, it holds that 
\begin{equation}\label{condlead}
(x_\ell-c_{\ell g})+\sum_{j\in \N_{\ell}(t)} \frac{2r^2}{(r^2-\| x_{\ell j}\|^2)^2}(x_\ell-x_j)=0.
\end{equation}
Denote by 
\begin{equation}\label{weight}
h_{ij}=\frac{2r^2}{(r^2-\| x_{ij}\|^2)^2}, \qquad \forall(i,\,j)\in E(t).
\end{equation}
We can construct a $N\times N$ matrix $H$ satisfying $H(i,i)=\sum_{j\in \mathcal{N}_i}h_{ij}$ and $H(i,j)=-h_{ij}$, where $i\neq j\in \N$. Since $x_{ij}\in [0,\,r-\varepsilon)$, $\forall (i,\,j)\in E(t)$, it holds that $h_{ij} > 0$. As shown in~\cite{ren2007multi}, $H$ is positive semidefinite with a single eigenvalue at the origin, of which the corresponding eigenvector is the unit column vector of length $N$, denoted by $\mathbf{1}_N$. By combining~\eqref{condflow} and~\eqref{condlead}, we get 
\begin{equation}\label{equilibrium1}
H\otimes I_2 \cdot \mathbf{x} + (\mathbf{x}-\mathbf{c}) =0
\end{equation}
where $\otimes$ denotes the Kronecker product~\cite{horn2012matrix}; $\mathbf{x}$ is the stack vector for $x_i$, $i\in \mathcal{N}$; $I_2$ is the $2\times 2$ identity matrix; $\mathbf{c}=\mathbf{1}_N\otimes c_{lg}$. Then 
$$H \otimes I_2\cdot  \mathbf{c}= (H \otimes I_2) \cdot (\mathbf{1}_N\otimes c_{lg})=(H \cdot \mathbf{1}_N)  \otimes (I_2 \cdot c_{lg}).$$
Since $H \cdot \mathbf{1}_N= \mathbf{0}_N$, it implies that $H \otimes I_2\cdot  \mathbf{c}=\mathbf{0}_{2N}$. By~\eqref{equilibrium1}, it implies that 
$
H\otimes I_2 \cdot (\mathbf{x}-\mathbf{c}) =0. 
$
Since we have shown that $H$ is positive semidefinite with one eigenvalue at the origin, \eqref{equilibrium1} holds only when $\mathbf{x}=\mathbf{c}$, i.e., $x_i=c_{\ell g}$, $\forall i\in \N$.

By LaSalle's Invariance principle~\cite{khalil2002nonlinear}, the closed-loop system under controller in Eq.~\eqref{control2} will converge to the largest invariant set inside the region 
\begin{equation}\label{invariant}
S=\{\mathbf{x}\in \mathbb{R}^{2N}\,|\, x_i=c_{\ell g}, \forall i\in \mathcal{N}\},
\end{equation}
as $t\rightarrow +\infty$. In other words, it means that all agents in $\N$ converge to the same point $c_{\ell g}$. Since clearly $c_{\ell g}\in R_{\ell g}$, by continuity all agents would enter $R_{\ell g}$ which has a minimal radius $r_{\min}$ by~\eqref{regions}. Consequently, there exists $\overline{t}<+\infty$ that $x_i(\overline{t})\in R_{\ell g}$, $\forall i\in \N$.

To conclude, given a connected initial graph $G(T_1)$ and the fixed leader $\ell \in \N$ for $t\geq T_1$, it is guaranteed that under controller in Eq.~\eqref{control2} all agents will converge to the region $R_{\ell g}$ in finite time. 
\end{proof}

\subsection{Progressive Goal and Leader Election}\label{goal-region}

To complete the solution to Problem~\ref{prob:main}, we discuss the election of the leader $\leader$ and the choice of a goal region $R_{\leader g}$ at time~$t$. As the first offline and fully decentralized step, we generate for each agent $i$ a \emph{high-level plan}, which is represented by the sequence of services that, if provided, guarantee the satisfaction of $\varphi_i$. Secondly, in a repetitive online procedure, each agent $i$ is assigned a value that, intuitivelly, represents the agent's urge to provide the next service in its high-level plan. Using ideas from bully leader election algorithm~\cite{garcia}, an agent with the strongest urge is always elected as a leader within the connectivity graph. By changing the urge dynamically at the times when services are provided, we ensure that each of the agents is elected as a leader infinitely often. Thus, each agent's precomputed high-level plan is followed.

\subsubsection{Offline high-level plan computation} \label{initial-plan}
Given an agent $i \in \N$, a set of services $\Sigma_i$, and an LTL formula $\varphi_i$ over $\Sigma_i$, a high-level plan for $i$ can be computed via standard model-checking methods as described in Section~\ref{sec:prelims}. Roughly, by translating $\varphi_i$ into a language equivalent B\"uchi automaton and by consecutive analysis of the automaton, a sequence of services $\plan_i = \sigma_{i1} \ldots \sigma_{ip_i}(\sigma_{ip_i+1} \ldots \sigma_{is_i})^\omega$, such that $\plan_i \models \varphi_i$ can be found. 

\subsubsection{Urge function} 
Let $i$ be a fixed agent, $t$ the current time and $\sigma_{i1}\ldots \sigma_{ik}$ a prefix of services of the high-level plan $\Omega_i$ that have been provided till $t$. Moreover, let $\tau_{i\last}$ denote the time, when the latest service, i.e., $\sigma_{i\last}=\sigma_{ik}$ was provided, or $\tau_{i\last}=0$ in case no service prefix of $\Omega_i$ has been provided, yet.

Using $\tau_{i\last}$, we could define agent $i$'s \emph{urge} at time $t$ as a tuple
\begin{equation}
\urge_i(t) = (t- \tau_{i\last} ,\, i).\label{eq:urge}
\end{equation}
Furthermore, to compare the agents' urges at time $t$, we use lexicographical ordering: $\urge_i(t) > \urge_j(t)$ if and only if 
\begin{itemize}
\item $t - \tau_{i\last} > t - \tau_{j\last}$, or 
\item $t - \tau_{i\last}$ = $t - \tau_{j\last}$, and $i > j$. 
\end{itemize}

Note that $i \neq j$ implies that $\urge_i(t) \neq \urge_j(t)$, for all $t\geq 0$. As a result, the defined ordering is a linear ordering and at any time $t$, there exists exactly one agent $i$ maximizing its urge $\urge_i(t)$.

\subsubsection{Overall algorithm}

The algorithm for an agent $i \in \N$ is summarized in Alg.~\ref{alg:main} and is run on each agent separately, starting at time $t = 0$.

\begin{algorithm}[!h]
\caption{Solution to Prob.~\ref{prob:main}}
\label{alg:main}
\begin{algorithmic}[1]
\small
\INPUT Agents' own ID $i$, the set of all agent IDs $\N$, formula $\varphi_i$
\OUTPUT $\trace_i$
\STATE compute plan $\Omega_i:= \sigma_{i1} \ldots \sigma_{ip_i}(\sigma_{ip_i+1} \ldots \sigma_{is_i})^\omega$\label{line:omega}
\STATE $\tau_{i\last} := 0$; $\sigma_{i\next} := \sigma_{i1}$
\label{line:init}
\STATE send $\mathsf{ready}(i)$ and wait to receive $\mathsf{ready}(j)$ for all $j \in \N \setminus \{i\}$\label{line:ready}
\IF {$i = N$} \label{line:elect1}
\STATE send $\mathsf{init\_elect}(i,\curr)$, where $\curr$ is the current time
\ENDIF\label{line:elect2}
\LOOP
\STATE wait to receive a message $m$\label{line:message}
\SWITCH {$m$}\label{line:switch}
\CASEONE {$m = \mathsf{init\_elect}(i',t)$ for some $i' \in \N$ and time $t$} \label{line:elect}
\STATE send $\mathsf{me}(\urge_i(t))$ and wait to receive  $\mathsf{me}(\urge_j(t))$ from all $j \in \N \setminus \{i\}$\label{line:electme}
\STATE elect the leader $\ell \in \N$ maximizing $\urge_\ell(t)$ \label{line:max}
\STATE send $\mathsf{finish\_elect}(i)$ and wait to receive  $\mathsf{finish\_elect}(j)$ from all $j \in \N \setminus \{i\}$\label{line:elected}
\IF {$\ell = i$} \label{line:leader}
\STATE $b_i := 1$ \label{line:15}
\STATE pick $R_{\ell g} = R_{ig}$, such that $\sigma_{i\next} \in L_i(R_{ig})$
\REPEAT
\STATE apply controller $u_i$ from Eq.~(\ref{control2}) \label{line:u1}
\UNTIL $x_j(t) \in R_{\ell g}$ for all $j \in \{i\} \cup \C_{i\next}$ \label{line:19}
\STATE  send $\mathsf{execute\_request}(\varpi_{ji\next})$ for all $j \in \C_{i\next}$\label{line:20}
\STATE execute $\pi_{i\next}$\label{line:21}
\STATE $\tau_{i\last}: = 0$; $\sigma_{i\next} := \sigma_{i\next +1}$\label{line:22}
\STATE  {update prefixes of $\TimeA,\Act_i,\TimeS$, and $\Ser_i$ }  \label{line:update1}
\STATE send $\mathsf{init\_elect}(i,\curr)$, where $\curr$ is the current time\label{line:23}
\ELSE\label{line:nonleader}
\STATE $b_i := 0$
\REPEAT
\STATE apply controller $u_i$ from Eq.~(\ref{control2})\label{line:u2}
\UNTIL a message $m$ is received; goto line~\ref{line:switch}\label{line:28}
\ENDIF\label{line:endif}
\CASE {$m = \mathsf{execute\_request}(\varpi_{ii'h'})$ for some $i' \in \N$, and $\sigma_{i'h'} \in \Sigma_{i'}$}
\STATE execute $\varpi_{ii'h'}$\label{line:32}
\STATE {update prefixes of $\TimeA$, and $\Act_i$}; goto line~\ref{line:switch}\label{line:update2}
\ENDSWITCH
\ENDLOOP
\end{algorithmic}
\end{algorithm}

The algorithm is initialized with the offline computation of the high-level plan $\Omega_i = \sigma_{i1} \ldots \sigma_{ip_i}(\sigma_{ip_i+1} \ldots \sigma_{is_i})^\omega$ as outlined above, and setting the values $\tau_{i\last} = 0$, $\sigma_{i\next} = \sigma_{i1}$ (lines~\ref{line:omega} -- \ref{line:init}). Then, the agent broadcasts a message to acknowledge the others that it is ready to proceed and waits to receive analogous messages from the remaining agents (line~\ref{line:ready}). The first leader election is triggered by a message sent by the agent $N$ (lines~\ref{line:elect1} -- \ref{line:elect2}) equipped with the time stamp $\curr$ of the current time.

Several types of messages can be received by the agent~$i$. Message $\mathsf{init\_elect}(i',t)$, where $i'$ is an arbitrary agent ID and $t$ is a time stamp, notifies that leader re-election is triggered (line \ref{line:elect}). In such a case, the agent sends out the message $\mathsf{me}(\urge_i(t))$ containing its own urge value $\urge_i(t)$ at the received time $t$ and waits to receive analogous messages from the others (line~\ref{line:electme}). The agent with the maximal urge is elected as the leader (line~\ref{line:max}) and the algorithm proceeds when each of the agents has set the new leader (line~\ref{line:elected}). Note that the elected leader is the same for all the agents.

The rest of the algorithm differs depending on whether the agent $i$ is the leader~(\ref{line:leader}--\ref{line:nonleader}) or not (\ref{line:nonleader}--\ref{line:endif}). The leader applies the controller from Eq.~(\ref{control2}) to reach a region where the next service $\sigma_{i \next}$ can be provided (lines~\ref{line:15}--\ref{line:19}). At the same time, it waits for the cooperating agents to reach the same region (line~\ref{line:19}). Then it provides service $\sigma_{i \next}$, with the help of the others (lines~\ref{line:20}--\ref{line:21}) and it sets the new latest service providing time $\tau_{i \last} = 0$, and the next service to be provided $\sigma_{i \next}$ to the following service in its plan, i.e. $ \sigma_{i \next +1}$, where, with a slight abuse of notation, we assume that $\sigma_{i s_i +1} = \sigma_{i p_i+1}$ (line~\ref{line:22}). For simplicity of presentation, we assume that the execution of an action $\pi_{i\next}$ is  synchronized with the execution of the action $\varpi_{ji\next}$, for all $j\in \C_{i\next}$. The details of the synchronization procedure are beyond the scope of this paper; for instance, the leader can decide a future time instance when $\varpi_{ji\next}$ should be executed and send it as a part of the $\mathsf{execute\_request}$ message.
Finally, the leader triggers a leader re-election (line~\ref{line:23}) with the current time $\curr$ as a time stamp.

A follower simply applies the controller from Eq.~(\ref{control2}) until it receives a message from the leader (lines~\ref{line:nonleader}--\ref{line:28}). The message can be either $\mathsf{execute\_request}(\varpi_{ii'h'})$ for a cooperating agent or $\mathsf{init\_elect}(i',t)$ for a non-cooperating agent.  

The algorithm naturally determines the trace $\trace_i = (\traj_i(t), \TimeA, \Act_i, \TimeS, \Ser_i)$ of the agent $i$ as follows: The trajectory $\traj_i(t)$ is given through the application of the controller $u_i$ from Eq.~(\ref{control2}) (lines~\ref{line:u1} and \ref{line:u2}). The sequences $\TimeA, \Act_i, \TimeS, \Ser_i$ are iteratively updated upon the agent's run (lines~\ref{line:update1}~and~\ref{line:update2}). Initially, they are all empty sequences and a time instant, an action, or a service is added to them whenever an action is executed (lines~\ref{line:21}~and~\ref{line:32}) or a service is provided (line~\ref{line:21}), respectively.

\bigskip

To prove that the proposed algorithm is correct, we first prove that each agent $i$ is elected as a leader infinitely many times:

\begin{lemma}\label{lemma:leader}
Given an agent $i\in \N$ at time $t$, there exists $T \geq t$, such that $\urge_i(T) > \urge_j(T)$, for all $j \in \N$, and $t \geq 0$.
\end{lemma}
\begin{proof}
Proof is given by contradiction. Assume the following:
\vspace{-0.3cm}
\begin{assumption}
For all $t'\geq t$ there exists some $j \in \N$, such that $\urge_i(t') < \urge_j(t')$.
\label{assump:1}
\end{assumption}
Consider that $\ell \in \N$ is set as the leader at time $t$, and an agent $i' \in \N$ maximizes $\urge_{i'}(t)$ among all agents in $\N$. Then, from the construction of Alg.~\ref{alg:main} and Lemmas~\ref{connected} and~\ref{convergence}, there exists $\tau_{\ell\next} \geq t$ when the next leader's desired service $\sigma_{\ell\next}$ has been provided and a leader re-election is triggered with the time stamp $\tau_{\ell\next}$.  Note that from Eq.~\eqref{eq:urge}, $\urge_{i'}(\tau_{\ell\next})$ is still maximal among the agents in $\N$,  and hence $i'$ becomes the next leader. Furthermore, there exists time $\tau_{i'\next} \geq \tau_{\ell\next}$ when the next desired service $\sigma_{i' \next}$ of agent $i'$ has been provided, and hence $\urge_{i'}(\tau_{i'\next} ) < \urge_{j}(\tau_{i'\next} )$, for all $j \in \N$, including the agent $i$. 

Since we assume that $i$ does not become a leader for any $t' \geq t$ (Assump.~\ref{assump:1}), it holds that $\urge_{i'}({t}'') < \urge_{i}({t}'')$ for all ${t}'' \geq \tau_{i'\next} $. Inductively, we can reason similarly about the remaining agents. As they are only finite number of agents $N$, after large enough $T\geq t$, we obtain that $\urge_{j}(t' ) < \urge_{i}(t')$ for all $j$ and for all $t' \geq T$. This contradicts Assump.~\ref{assump:1} and hence the proof is complete.
\end{proof}

\medskip

From~Lemmas~\ref{connected},~\ref{convergence}, and \ref{lemma:leader}, the correctness of the high-level plan computation (proven~\eg in~\cite{principles}), and the construction of Alg.~\ref{alg:main}, we obtain, that $\phi_i$ is satisfied for all $i \in \N$.

From Alg.~\ref{alg:main}, each agent synthesizes its high-level plan $\Omega_i$ and waits for the first leader to be elected. 
Denote the first leader by $\ell_1\in \N$. 
Lemma~\ref{convergence} guarantees that there exists a finite time $\overline{t}_1>0$ that $x_{\ell_1}(\overline{t}_1)\in R_{\ell_1 \next}$, while at the same time Lemma~\ref{connected} ensures that the communication network $G(t)$ remains connected, $\forall t\in [0,\overline{t}_1]$.
At $\tau_{\ell_1\next} \geq \overline{t}_1$, {the first service $\sigma_{\ell_1 \next}$ of the leader's high-level plan is provided, defining a prefix of the agent $i$'s trace as $\TimeA(1) = \tau_{\ell_1\next}$, $\Act_i(\tau_{\ell_1\next}) = \pi_{i\next}$, $\TimeS(1) = \tau_{\ell_1\next}$, $\Ser_i(\tau_{\ell_1\next}) = \sigma_{i\next}$. Furthermore, $\mathbb{T}^A_j (1) = \tau_{\ell_1\next}$, $\Act_j(\tau_{\ell_1\next}) = \varpi_{ji\next}$.} 
Afterwards, a new leader $\ell_2\in \N$ is elected according to Alg.~\ref{alg:main} and $R_{\ell_2 \next}$ is set as the goal region. Now the controller from Eq.~(\ref{control2}) is switched to the case when $\ell_2$ is the leader. By induction, we obtain 
that for all $t \geq 0$ it holds
\begin{itemize}
\item Given a leader $\ell_t$ and a goal region $R_{\ell\next}$ at time $t$, there exists $\bar{t} \geq t$, when $x_\ell(\bar{t}) \in R_{\ell \next}$.
\item $G(t)$ is connected.
\end{itemize}
Together with Lemma~\ref{lemma:leader}, we conclude that $\phi_i$ is satisfied for all $i \in \N$.

\begin{corollary}
The proposed solution in Alg.~\ref{alg:main} is a solution to Problem~\ref{prob:main}.
\end{corollary}


\section{Example}
\label{sec:example}
In the following case study, we present an illustrative example of a team of four autonomous robots with heterogeneous functionalities and capacities. The proposed algorithms are implemented in Python 2.7. All simulations are carried out on a desktop computer (3.06 GHz Duo CPU and 8GB of RAM).

\subsection{System Description}

Denote by the four autonomous agents $\mathfrak{R}_1$, $\mathfrak{R}_2$, $\mathfrak{R}_3$ and $\mathfrak{R}_4$. They all satisfy the dynamics specified by~\eqref{dynamics}. They all have the communication radius $1.5m$, while $\varepsilon$ is chosen to be $0.1m$. The workspace of size $4m\times 4m$ is given in Figure~\ref{workspace}, within which the regions of interest for $\mathfrak{R}_1$ are $R_{11}$, $R_{12}$ (in red), for $\mathfrak{R}_2$ are $R_{21}$, $R_{22}$ (in green), for $\mathfrak{R}_3$ are $R_{31}$, $R_{32}$ (in blue) and for $\mathfrak{R}_4$ are $R_{41}$, $R_{42}$ (in cyan). 

Besides the motion among these regions, each agent can  provide various services as described in the following: 
agent $\mathfrak{R}_1$ can load ($l_H,l_A$), carry and unload ($u_H,u_A$) a heavy object \textsf{H} or a light object \textsf{A}. Besides, it can help agent $\mathfrak{R}_4$ to assemble ($h_C$) object \textsf{C} ; 
agent $\mathfrak{R}_2$ is capable of helping the agent $\mathfrak{R}_1$ to load the heavy object \textsf{H} ($h_H$), and to execute two simple tasks ($t_1$, $t_2$) without help from others; 
agent $\mathfrak{R}_3$ is capable of taking snapshots  ($s$) when being present in its own or others' goal regions; 
agent $\mathfrak{R}_4$ can assemble ($a_C$) object \textsf{C} under the help of agent $\mathfrak{R}_1$.



\subsection{Task Description}
Each agent within the team is locally-assigned complex tasks that require collaboration: 
agent $\mathfrak{R}_1$ has to periodically load the heavy object \textsf{H} at region $R_{11}$, unload it at region $R_{12}$, load the light object \textsf{A} at region $R_{12}$, unload it at region $R_{11}$. In LTL formula, it is specified as 
\begin{align*}
\phi_1 = & \ \Always \Event \big( (l_H\wedge h_H \wedge r_{11}) \wedge \Next (u_H\wedge r_{12})\big) \wedge \\ & \ \Always \Event \big( (l_A \wedge r_{12}) \wedge (u_A \wedge r_{11})\big);
\end{align*}
Agent $\mathfrak{R}_2$ has to service the simple task $t_1$ at region $R_{21}$ and task $t_2$ at region $R_{22}$ in sequence, but it requires $\mathfrak{R}_2$ to witness the execution of task $t_2$, by taking a snapshot at the moment of the execution. It is specified as 
\begin{align*}
\phi_2= \Event \big( ( t_1 \wedge r_{21}) \wedge \Event (t_2\wedge s \wedge r_{22})\big);
\end{align*}
Agent $\mathfrak{R}_3$ has to surveil over  both of its goal regions ($R_{31}$, $R_{32}$) and take snapshots there, which is specified as \begin{align*}
\phi_3=\Always \Event (s \wedge r_{31}) \wedge \Always \Event (s \wedge r_{32});
\end{align*}
Agent $\mathfrak{R}_4$ has to assemble object \textsf{C} at its goal regions ($R_{41}$, $R_{42}$) infinitely often, which is specified as \begin{align*}
\phi_4=\Always \Event (a_C \wedge r_{41}) \wedge \Always \Event (a_C \wedge r_{42}).
\end{align*}
Note that tasks $\phi_1$, $\phi_3$ and $\phi_4$ require the collaboration task be performed infinitely often. 

\subsection{Simulation Results}
Initially, the agents start evenly from the $x$-axis, i.e., $(0,\,0),\,(1.3,\,0), \, (2.6,\,0), \, (3.9,\,0)$. By Definition~\ref{edge}, the initial edge set is $E(0)=\{(1,\,2), \,(2,\,3),\,(3,\,4)\}$, yielding a connected $G(0)$.

The system is simulated for $35s$, of which the video demonstration can be viewed here~\cite{video}.
In particular, when the system starts, each agent synthesizes its local plan as described in Section~\ref{initial-plan}. After running the leader election scheme proposed in Section~\ref{goal-region}, agent $\mathfrak{R}_1$ is chosen as the leader. As a result, controller~\eqref{control} is applied for $\mathfrak{R}_1$ as the leader and the rest as followers, while the next goal region of $\mathfrak{R}_1$ is $R_{11}$. As shown by Theorem~\ref{convergence}, all agents belong to $R_{41}$ after $t=3.8s$. After that agent $\mathfrak{R}_2$ helps agent $\mathfrak{R}_1$ to load object \textsf{H}. Then agent $\mathfrak{R}_2$ is elected as the leader after collaboration is done, where $R_{21}$ is chosen as the next goal region. At $t=6.1s$, all agents converge to $R_{21}$. Afterwards, the leader and goal region is switched in the following order: 
$\mathfrak{R}_3$ as leader to region $R_{31}$ at $t=6.1s$; 
$\mathfrak{R}_4$ as leader to region $R_{41}$ at $t=8.1s$; 
$\mathfrak{R}_4$ as leader to region $R_{42}$ at $t=10.6s$; 
$\mathfrak{R}_2$ as leader to region $R_{22}$ at $t=14.2s$; 
$\mathfrak{R}_3$ as leader to region $R_{32}$ at $t=16.3s$; 
$\mathfrak{R}_1$ as leader to region $R_{12}$ at $t=18.2s$; 
$\mathfrak{R}_1$ as leader to region $R_{11}$ at $t=20.1s$; 
$\mathfrak{R}_3$ as leader to region $R_{31}$ at $t=24.2s$; 
$\mathfrak{R}_3$ as leader to region $R_{32}$ at $t=25.7s$; 
$\mathfrak{R}_4$ as leader to region $R_{41}$ at $t=28.1s$; 
$\mathfrak{R}_4$ as leader to region $R_{42}$ at $t=31.4s$. The above arguments are summarized in Table~\ref{tab:leader}.

Figure~\ref{workspace} shows the snapshot of the simulation at time $t=11.2s$, when agent $\mathfrak{R}_4$ was chosen as the leader and $R_{42}$ as the goal region.  Figure~\ref{full-traj} shows the trajectory of $\mathfrak{R}_1$, $\mathfrak{R}_2$,$\mathfrak{R}_3$, $\mathfrak{R}_4$ during time $[0, \, 34.7s]$, in red, green, blue, cyan respectively. Furthermore, the pairwise distance for neighbours within $E(0)$ is shown in Figure~\ref{full-traj}. It can be verified that they stay below the constrained radius $1.5m$ thus the agents remain connected.


\begin{figure}[t]
\begin{minipage}[t]{0.49\linewidth}
\centering
\includegraphics[height =1\textwidth]{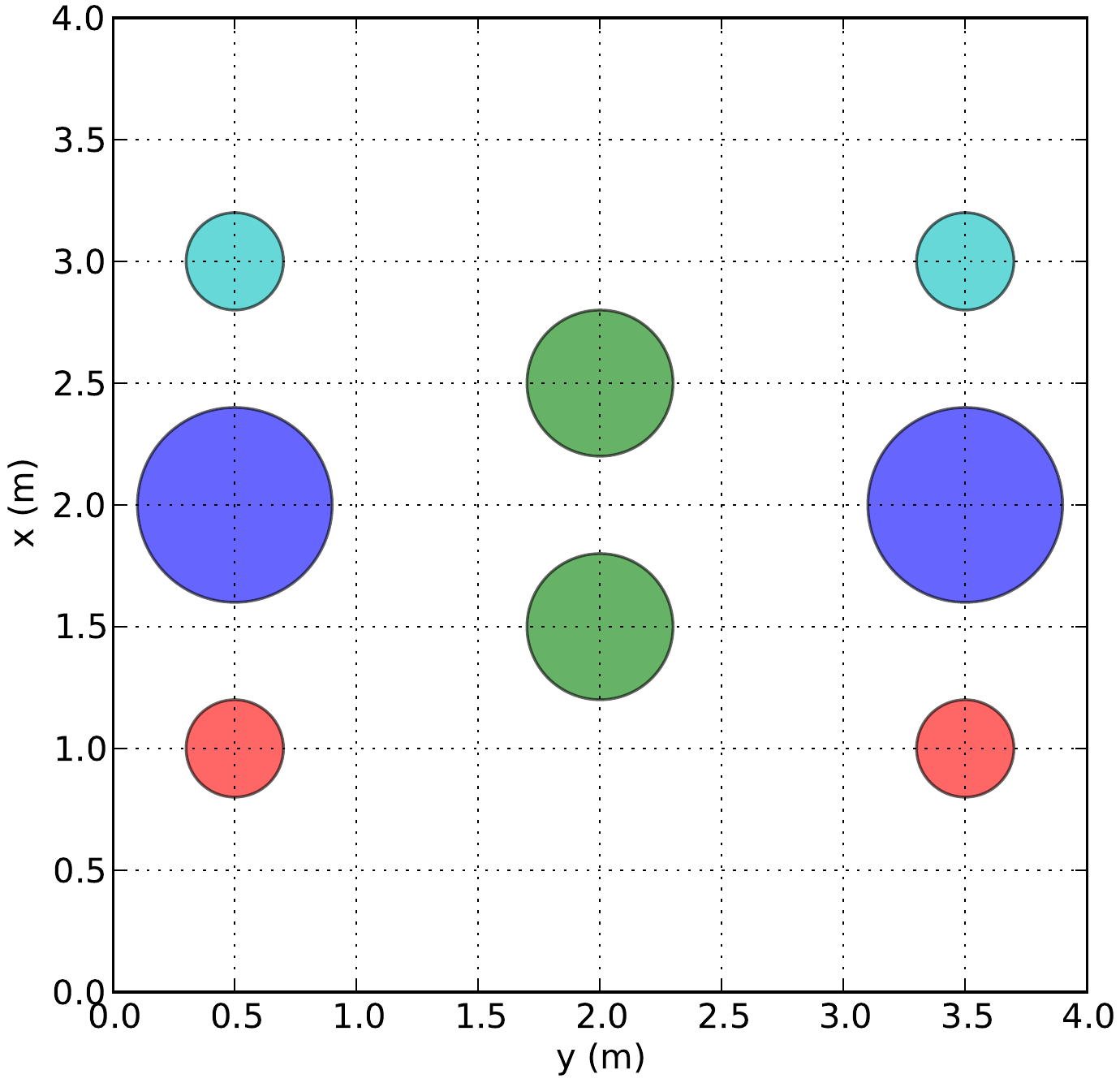}

  \end{minipage}
\begin{minipage}[t]{0.5\linewidth}
\centering
    \includegraphics[height =0.98\textwidth]{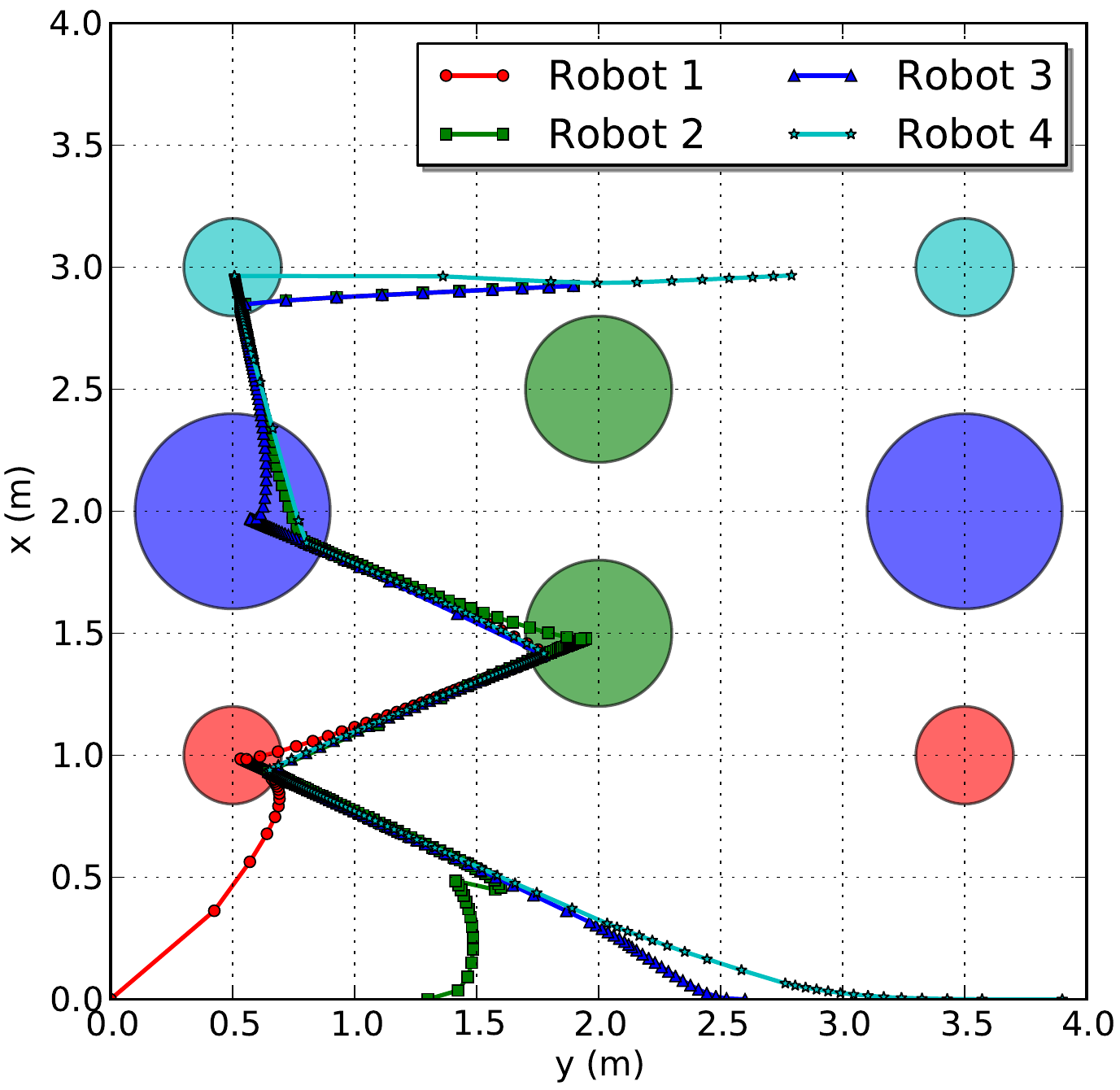}
\end{minipage}
\caption{Left: the workspace structure, where the goal regions for each agent are indicated by color; Right: snapshot of simulation at $t=11.2s$.}
\label{workspace}
\end{figure}

\begin{table}[t]
\centering

\scalebox{0.9}{

\centering
    \begin{tabular}{|c||c|c|c|c|}
        \hline  
        \textbf{Time} ($s$)  & $(0,3.8)$ &    $(3.8,6.1)$ & $(6.1,8.1)$ & $(8.1,10.6)$ \\        
        \hline
        Leader  & $\mathfrak{R}_1$ &    $\mathfrak{R}_2$ & $\mathfrak{R}_3$ & $\mathfrak{R}_4$\\
        \hline
        Goal Region  &$R_{11}$ &    $R_{21}$ & $R_{31}$ & $R_{41}$\\
        \hhline{|=||=|=|=|=|}
        \textbf{Time} ($s$)  & $(10.6,14.2)$ &    $(14.2,16.3)$ & $(16.3,18.2)$ & $(18.2,20.1)$ \\        
        \hline
        Leader  & $\mathfrak{R}_4$ &    $\mathfrak{R}_2$ & $\mathfrak{R}_3$ & $\mathfrak{R}_1$\\
        \hline
        Goal Region  &$R_{42}$ &    $R_{22}$ & $R_{32}$ & $R_{12}$\\
        \hhline{|=||=|=|=|=|}
        \textbf{Time} ($s$)  & $(20.1,24.2)$ &    $(24.2,25.7)$ & $(25.7,28.1)$ & $(28.1,31.4)$ \\        
        \hline
        Leader  & $\mathfrak{R}_1$ &    $\mathfrak{R}_3$ & $\mathfrak{R}_3$ & $\mathfrak{R}_4$\\
        \hline
        Goal Region  &$R_{11}$ &    $R_{31}$ & $R_{32}$ & $R_{41}$\\
        \hline
    \end{tabular}
  }
\caption{Leader Eelection Scheme}
\label{tab:leader}  
\end{table}

\begin{figure}[t]
\begin{minipage}[t]{0.49\linewidth}
\centering
\includegraphics[height =1\textwidth]{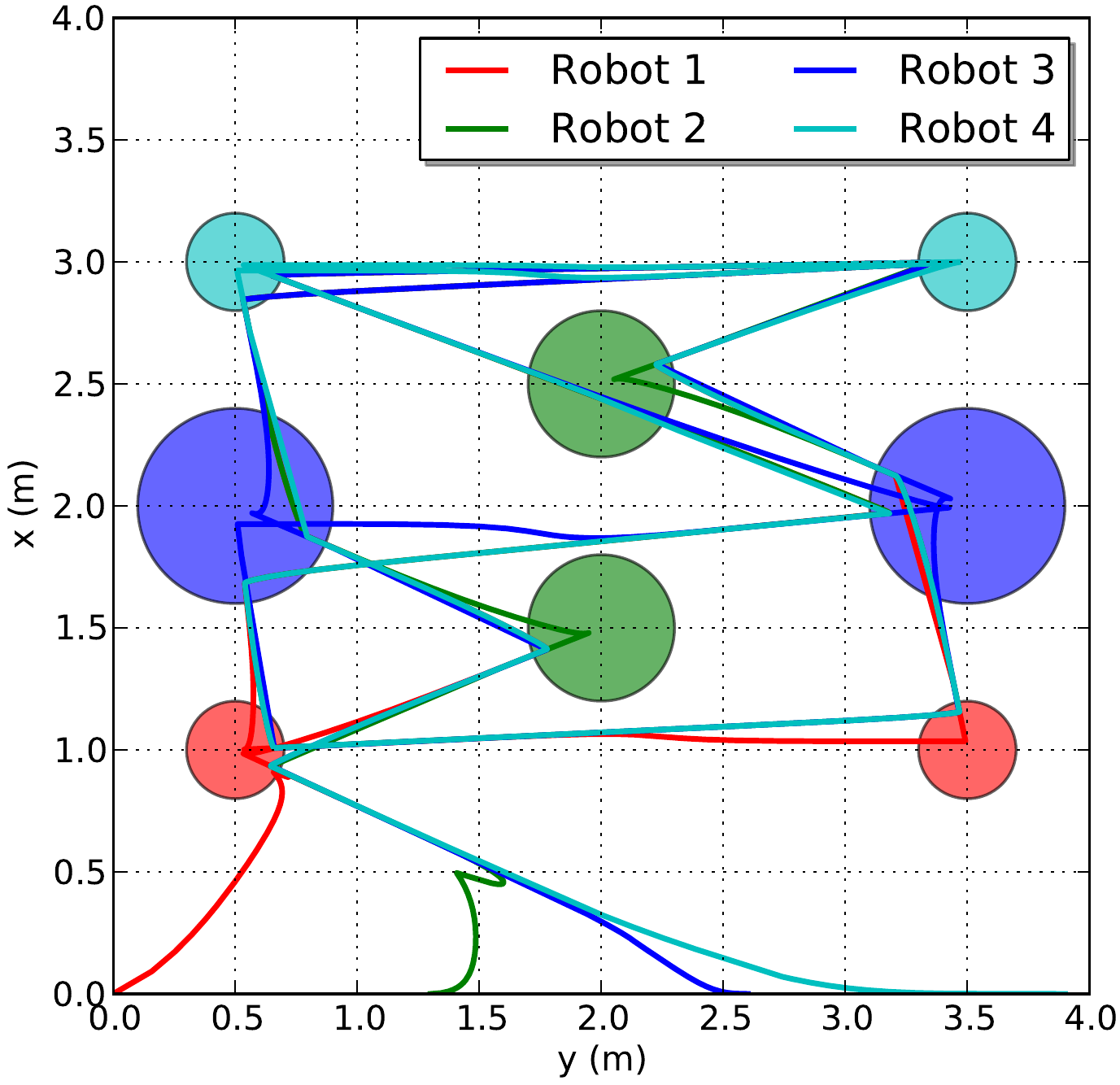}

  \end{minipage}
\begin{minipage}[t]{0.5\linewidth}
\centering
    \includegraphics[height =0.98\textwidth, width =0.98\textwidth]{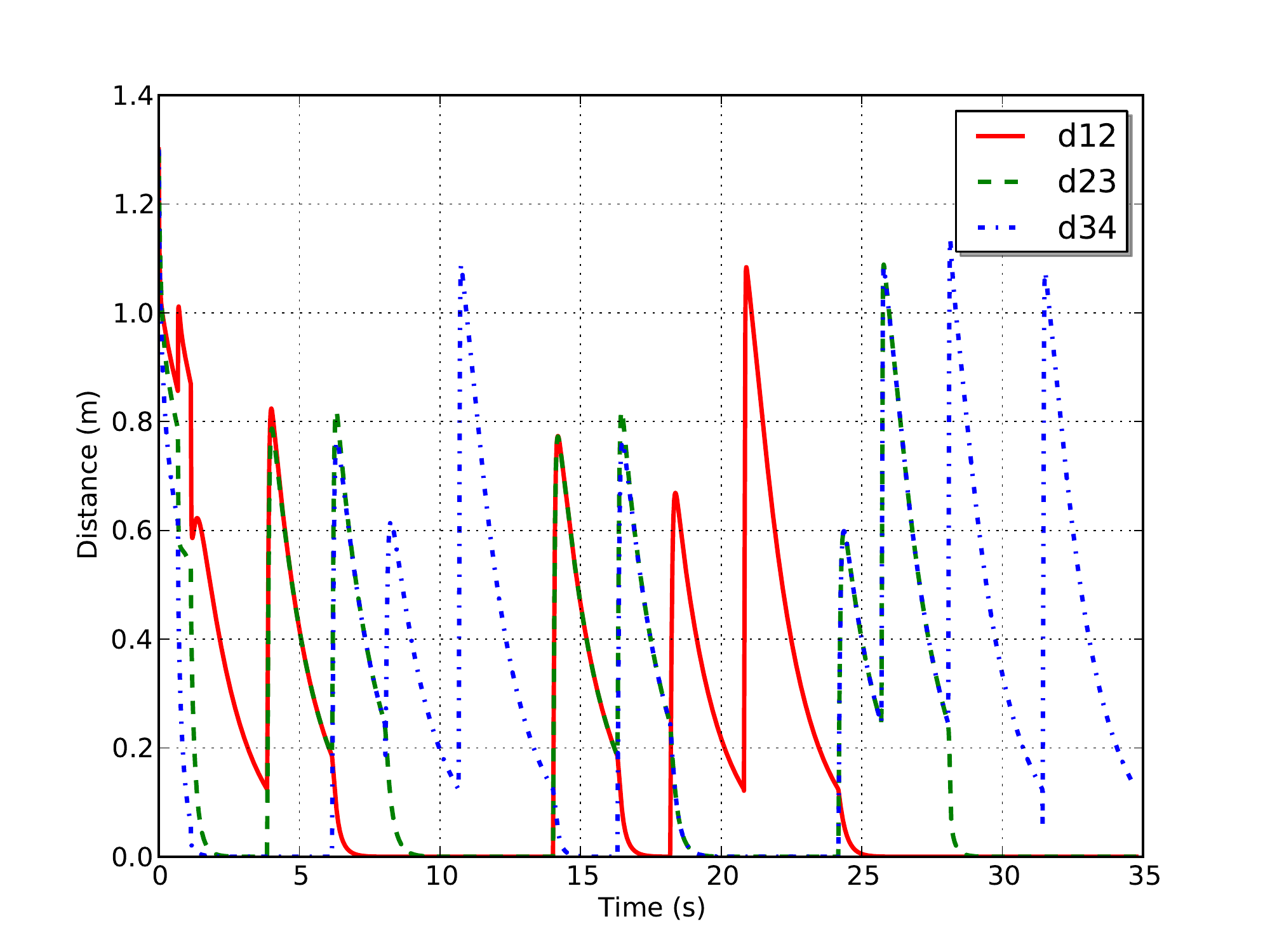}
\end{minipage}
\caption{Left: the agents' trajectory during time $[0, \, 34.8s]$; Right: the evolution of pair-wise distances $\|x_{12}\|,\|x_{23}\|, \|x_{34}\|$, which all stay below the communication radius $1.5m$ as required by the connectivity constraints.}
\label{full-traj}
\end{figure}

\section{Conclusions and Future Work}
\label{sec:conc}
We present a distributed motion and task control framework for multi-agent systems under complex local LTL tasks and connectivity constraints. It is guaranteed that all individual tasks including both local and collaborative services are fulfilled, while at the same time connectivity constraints are satisfied. Further work includes inherently-coupled dynamics and time-varying network topology. 

\bibliographystyle{plain}
\bibliography{refer}

\begin{thebibliography}{10}

\bibitem{principles}
C.~Baier and J.-P. Katoen.
\newblock {\em Principles of Model Checking}.
\newblock MIT Press, 2008.

\bibitem{yushan-tr2012}
Y.~Chen, X.~C. Ding, A.~Stefanescu, and C.~Belta.
\newblock Formal approach to the deployment of distributed robotic teams.
\newblock {\em IEEE Transactions on Robotics}, 28(1):158--171, 2012.

\bibitem{video}
Demonstration.
\newblock \url{https://www.dropbox.com/s/yhyueagokeihv7k/simulation.avi}.

\bibitem{1641830}
D.V. Dimarogonas and K.J. Kyriakopoulos.
\newblock A connection between formation control and flocking behavior in
  nonholonomic multiagent systems.
\newblock In {\em Robotics and Automation, 2006. ICRA 2006. Proceedings 2006
  IEEE International Conference on}, pages 940--945, May 2006.

\bibitem{egerstedt2001formation}
M.~B. Egerstedt and X.~Hu.
\newblock Formation constrained multi-agent control.
\newblock 2001.

\bibitem{dimos-cdc12}
I.~Filippidis, D.V. Dimarogonas, and K.J. Kyriakopoulos.
\newblock Decentralized multi-agent control from local ltl specifications.
\newblock In {\em Proceedings of the {IEEE} Conference on Decision and Control
  (CDC)}, pages 6235--6240, 2012.

\bibitem{garcia}
H.~Garcia-Molina.
\newblock Elections in a distributed computing system.
\newblock {\em IEEE Transactions on Computers}, C-31(1):48--59, 1982.

\bibitem{ltl2ba}
P.~Gastin and D.~Oddoux.
\newblock {LTL2BA} tool, viewed September 2012.
\newblock URL: http://www.lsv.ens-cachan.fr/~gastin/ltl2ba/.

\bibitem{meng-cdc2013}
M.~Guo and D.~V. Dimarogonas.
\newblock Reconfiguration in motion planning of single- and multi-agent systems
  under infeasible local {LTL} specifications.
\newblock 2013.
\newblock To appear.

\bibitem{hong2006tracking}
Y.~Hong, J.~Hu, and L.~Gao.
\newblock Tracking control for multi-agent consensus with an active leader and
  variable topology.
\newblock {\em Automatica}, 42(7):1177--1182, 2006.

\bibitem{horn2012matrix}
R.~A. Horn and C.~R. Johnson.
\newblock {\em Matrix analysis}.
\newblock Cambridge university press, 2012.

\bibitem{sertac-ijnc2010}
S.~Karaman and E.~Frazzoli.
\newblock Vehicle routing with temporal logic specifications: Applications to
  multi-{UAV} mission planning.
\newblock {\em International Journal of Robust and Nonlinear Control},
  21:1372--1395, 2011.

\bibitem{khalil2002nonlinear}
H.~K. Khalil and J.~W. Grizzle.
\newblock {\em Nonlinear systems}, volume~3.
\newblock Prentice hall Upper Saddle River, 2002.

\bibitem{marius-cdc2011}
M.~Kloetzer, X.~C. Ding, and C.~Belta.
\newblock Multi-robot deployment from ltl specifications with reduced
  communication.
\newblock In {\em Proceedings of the {IEEE} Conference on Decision and Control
  and European Control Conference}, pages 4867--4872, 2011.

\bibitem{loizou-cdc2005}
S.~G. Loizou and K.~J. Kyriakopoulos.
\newblock Automated planning of motion tasks for multi-robot systems.
\newblock In {\em Proceedings of the {IEEE} Conference on Decision and Control
  (CDC)}, volume~44, December 2005.

\bibitem{quo-icra2004}
M.~M. Quottrup, T.~Bak, and R.~I. Zamanabadi.
\newblock Multi-robot planning : a timed automata approach.
\newblock In {\em Proceedings of the {IEEE} International Conference on
  Robotics and Automation (ICRA)}, volume~5, pages 4417--4422, 2004.

\bibitem{ren2007multi}
W.~Ren.
\newblock Multi-vehicle consensus with a time-varying reference state.
\newblock {\em Systems \& Control Letters}, 56(7):474--483, 2007.

\bibitem{1470239}
W.~Ren, R.~W. Beard, and E.~M. Atkins.
\newblock A survey of consensus problems in multi-agent coordination.
\newblock In {\em American Control Conference, 2005. Proceedings of the 2005},
  pages 1859--1864 vol. 3, June 2005.

\bibitem{jana-acc2014}
J.~Tumova and D.~Dimarogonas.
\newblock A receding horizon approach to multi-agent planning from ltl
  specifications.
\newblock In {\em Proceedings of the American Control Conference (ACC)}, 2014.
  To appear.

\bibitem{alphan-ijrr2013}
A.~Ulusoy, S.~L. Smith, X.~C. Ding, C.~Belta, and D.~Rus.
\newblock Optimality and robustness in multi-robot path planning with temporal
  logic constraints.
\newblock {\em International Journal of Robotics Research}, 32(8):889--911,
  2013.

\bibitem{lygeros-ecc2013}
C.~Wiltsche, F.~A. Ramponi, and J.~Lygeros.
\newblock Synthesis of an asynchronous communication protocol for search and
  rescue robots.
\newblock pages 1256--1261, 2013.

\end{thebibliography}

\end{document}